\documentclass[prl,twocolumn,showpacs,groupedaddress,notitlepage,superscriptaddress,floatfix,nofootinbib]{revtex4-2}

\usepackage[english]{babel}
\usepackage{amsmath}
\usepackage{amssymb}
\usepackage{amsthm}
\usepackage{amsfonts}
\usepackage{braket}
\usepackage{cases}
\usepackage{color}
\usepackage{xcolor}
\usepackage{multirow}
\usepackage{float}
\usepackage{bm}
\usepackage{hyperref}
\hypersetup{
    colorlinks=true, 
    linkcolor=cyan,
    citecolor=magenta, 
    filecolor=magenta, 
    urlcolor=cyan,
    runcolor=cyan
}
\usepackage[capitalise]{cleveref}
\usepackage{booktabs} 

\newtheorem{mytheorem}{Theorem}
\newtheorem{mylemma}{Lemma}

\newtheorem{myproposition}{Proposition}

\newtheorem{mydefinition}{Definition}

\newtheorem{myconjecture}{Conjecture}

\newcommand{\ignore}[1]{}

\newcommand{\mc}[1]{\mathcal{#1}}

\usepackage{commath}
\usepackage{adjustbox}
\usepackage[normalem]{ulem}

\newcommand{\bes} {\begin{subequations}}
\newcommand{\ees} {\end{subequations}}
\newcommand{\beq}{\begin{equation}}
\newcommand{\eeq}{\end{equation}}
\renewcommand{\norm}[1]{\| #1 \|}

\def\a{\alpha}
\def\b{\beta}

\def\t{\theta}

\def\D{\Delta}

\def\>{\rangle}
\def\<{\langle}
\def\Tr{\mathrm{Tr}}
\def\Pr{\mathrm{Pr}}
\newcommand{\ketbra}[1]{|{#1}\>\!\<#1|}

\def\mE{\mathcal{E}}
\def\mF{\mathcal{F}}
\def\mG{\mathcal{G}}
\def\mH{\mathcal{H}}
\def\mI{\mathcal{I}}

\def\mT{\mathcal{T}}
\def\mU{\mathcal{U}}
\def\mV{\mathcal{V}}
\def\mZ{\mathcal{Z}}

\newcommand{\HX}{H_X}
\newcommand{\RX}{G_X}
\newcommand{\HI}{H_{\mathrm{I}}}
\newcommand{\tHI}{\widetilde{H}_{\mathrm{I}}}
\newcommand{\HTFI}{H_{\mathrm{TFI}}}
\newcommand{\UDAD}{U_{\mathrm{tot}}(T)}
\newcommand{\PDAD}{P_{U_\mathrm{tot}}}
\newcommand{\UpDAD}{U'_{\mathrm{tot}}(T)}
\newcommand{\PpDAD}{P_{U'_\mathrm{tot}}}
\newcommand{\UIQP}{U_{\mathrm{IQP}}}
\newcommand{\PIQP}{P_{U_{\mathrm{IQP}}}}
\newcommand{\UpIQP}{U'_{\mathrm{IQP}}}
\newcommand{\tU}{\widetilde{U}}
\newcommand{\UL}{U_{\mathrm{L}}}
\newcommand{\UR}{U_{\mathrm{R}}}

\DeclareMathOperator{\sech}{sech}

\begin{document}

\title{Digital-Analog-Digital Quantum Supremacy}
\author{Daniel A. Lidar}
\affiliation{Departments of Electrical \& Computer  Engineering, Chemistry, and Physics \& Astronomy, University of Southern California, Los Angeles, 
CA 90089}
\affiliation{Center for Quantum Information Science \& Technology, University of
Southern California, Los Angeles, CA 90089}

\begin{abstract}
Quantum supremacy has been explored extensively in gate-model settings. Here, we introduce a quantum-supremacy framework for a hybrid digital-analog-digital quantum computing (DADQC) model. We consider a device that applies an initial layer of single-qubit gates, a single transverse-field Ising analog block, and a final single-qubit layer before $Z$-basis readout. The analog block approximates $Z$-diagonal Ising evolution, and we prove that the resulting output distribution is within constant total-variation (TV) distance of an Instantaneous Quantum Polynomial-time (IQP) circuit. Our bounds and constructions are established for fully connected as well as bounded-degree hardware graphs, matching a variety of architectures, including trapped-ion, neutral atom, and superconducting platforms. Assuming anticoncentration (which we prove for all-to-all hardware graphs and conjecture for bounded-degree hardware graphs) and an average-case hardness conjecture for the associated complex-temperature Ising partition functions, standard reductions imply that any efficient classical sampler achieving constant TV error collapses the polynomial hierarchy. Our results imply that quantum-supremacy tests are possible on today's quantum annealers, as well as other devices capable of hybrid digital-analog quantum evolution.
\end{abstract}

\maketitle

The pursuit of quantum computational supremacy, i.e., the demonstration of a task efficiently solvable by a quantum device but intractable for classical  computers~\cite{Preskill:2012aa}, has become a central benchmark for near-term quantum hardware. Following the theoretical proposals of boson sampling~\cite{Aaronson2011} and Instantaneous Quantum Polynomial-time (IQP) circuits~\cite{Bremner2010,Bremner:2016aa}, experimental efforts have realized quantum advantage on several distinct platforms. These include random circuit sampling (RCS) on superconducting qubits~\cite{Arute2019, zhuQuantumComputationalAdvantage2022,morvan2023phase,Gao:2025aa} and trapped ions~\cite{DeCross:2025aa}, as well as Gaussian boson sampling in photonic networks~\cite{Zhong:2020,Madsen2022,Deng:2023}. Each of these schemes leverages the hardness of approximately sampling from the output distribution of a nonuniversal but physically natural quantum process, under plausible complexity-theoretic assumptions.
A complementary route is the Quantum Approximate Optimization Algorithm (QAOA), a shallow alternating circuit for combinatorial optimization~\cite{farhi2014quantum}, whose lowest-depth instances have classically intractable output distributions under standard assumptions~\cite{FarhiHarrow-QAOA,KroviQAOAAvgCase}.

These results rely on discrete gate sequences and specific random ensembles (e.g., Haar-random unitaries or random interferometers). An alternative route to supremacy is through analog Hamiltonian dynamics, where complexity arises from continuous-time evolution under restricted Hamiltonians. This paradigm includes stoquastic adiabatic quantum computation, which uses Hamiltonians whose off-diagonal matrix elements are all real and non-positive in the computational basis~\cite{Bravyi:2009sp} and has been extensively reviewed~\cite{Albash-Lidar:RMP,Hauke:2019aa}; provable quantum speedups are known in some cases~\cite{Somma:2012kx,hastings2020power,Gilyen:2020aa}. The complexity-theoretic status of stoquastic dynamics and their potential for demonstrating supremacy remain unresolved. Recent progress in quantum annealing~\cite{King:2019aa,King:22,king2022coherent,Ebadi:22,bauza2024scaling}, including evidence of beyond-classical quantum simulation~\cite{king2024computationalsupremacyquantumsimulation}, motivates revisiting this question.

Here, we consider a minimal hybrid model that bridges the digital and analog paradigms for superconducting flux qubits~\cite{DeshpandeAQC2025}, and applies directly to other modalities that implement transverse-field Ising dynamics, such as Rydberg atoms~\cite{Glaetzle:2017aa,Bernien:2017aa,Scholl:2021aa,Ebadi:22}, trapped ions~\cite{Kim:2010aa,Islam:2013mi,Monroe:2021aa}, and superconducting-cavity platforms~\cite{Lamata:2017aa}. The device applies an initial layer of single-qubit gates, then a single analog block of transverse-field Ising evolution, and finally another single-qubit layer. Although the analog block is stoquastic in the computational basis, the terminal single-qubit layer applies random 
$XY$-plane product rotations, so the overall sampling task is not constrained by stoquasticity. We call this model ``digital-analog-digital quantum computing'' (DADQC). We show that its output distribution is within constant total-variation distance of an IQP circuit, thereby inheriting standard IQP sampling-hardness assumptions~\cite{Bremner:2016aa,Boixo2018,Hangleiter2018}, one of which (anticoncentration) we prove for all-to-all connected hardware graphs. We do not require adiabaticity; DADQC is intermediate between adiabatic quantum annealing as originally conceived~\cite{kadowaki_quantum_1998} and the circuit model, in the spirit of diabatic quantum annealing~\cite{crosson2020prospects}. 

Our result situates DADQC alongside RCS, boson sampling, and shallow alternating circuits as a viable route to supremacy~\cite{RevModSupremacy2023}. It shows that a model that is fully compatible with current superconducting quantum annealers~\cite{DeshpandeAQC2025} and other modalities capable of hybrid digital-analog QC~\cite{Glaetzle:2017aa,Bernien:2017aa,Scholl:2021aa,Ebadi:22,Kim:2010aa,Islam:2013mi,Monroe:2021aa,Lamata:2017aa} can exhibit sampling hardness equivalent to that of IQP circuits. This extends quantum supremacy beyond the present circuit-model setting, providing a complementary theoretical foundation compatible with analog hardware.

\textit{Model}.---%
Consider $n$ qubits occupying the vertex set $\mV_{\mH}$ of a fixed ``hardware graph'' $\mH_n=(\mV_{\mH},\mE_{\mH})$, whose edges (couplers) specify native two-qubit interactions. 
Let $\UL=\bigotimes_{i\in\mV_{\mH}}U_{\mathrm{L},i}$ and $U_{\mathrm{R}}=\bigotimes_{i\in\mV_{\mH}}U_{\mathrm{R},i}$ be arbitrary single-qubit layers. We assume initialization in $\ket{0}^{\otimes n}$ and a $Z$-basis measurement after $U_{\mathrm{R}}$. Denoting by $U_{\mathrm{A}}(T)$ the analog unitary, the device implements the total evolution
\beq
\label{eq:full}
\UDAD = U_{\mathrm{R}} U_{\mathrm{A}}(T) \UL,\quad U_{\mathrm{A}}(T) = \mT e^{-i \int_0^T \HTFI(t) dt} ,
\eeq
where $\mT$ denotes time ordering. The analog evolution is
generated by the transverse field Ising model Hamiltonian
\beq
\label{eq:HTFI}
\HTFI(t) \equiv A(t)\HX + B(t)\HI ,
\eeq
with transverse field $\HX\equiv\sum_{i\in\mV_{\mH}}X_i$ and Ising term
\beq
\label{eq:HX-HI}
\HI \equiv  \sum_{i\in \mV_{\mH}} h_i Z_i + \sum_{(i,j)\in \mE_{\mH}} J_{ij} Z_i Z_j ,
\eeq
where $X$ and $Z$ denote the Pauli matrices.
The local fields $h_i$ and coupling constants $J_{ij}$ are dimensionless; the annealing schedules $A(t),B(t)$ have units of energy (we set $\hbar\equiv 1$) and are smooth, with $A(t)\ge 0$ monotonically decreasing toward $0$ and $B(t)\ge 0$ monotonically increasing on $[0,T]$. 

For comparison, traditional IQP circuits~\cite{Bremner2010,Bremner:2016aa} have the form
\begin{subequations}
\label{eq:IQP-general}
\begin{align}
\label{eq:IQP-U}
\UIQP&= W^{\otimes n} e^{-iH_Z} W^{\otimes n},\\
\label{eq:IQP-HZ}
H_Z&=\sum_{i\in\mV_{\mG}} v_i Z_i  +  \sum_{(i,j)\in\mE_{\mG}} w_{ij} Z_i Z_j,
\end{align}
\end{subequations}
on an ``interaction graph'' $\mG=(\mV_{\mG},\mE_{\mG})$, where $v_i,w_{ij}\in\{\frac{\pi}{8}k:k=0,\dots,7\}$ (the ``$\pi/8$ grid'') and $W$ is the Hadamard transform. Let $s\in\{0,1\}^n$ denote a bitstring; the outcome distribution of the IQP model after initialization in $\ket{0^{\otimes n}}$ is
\begin{equation}
\label{eq:IQP-P}
\PIQP(s)  =  \bigl|\bra{s}\UIQP\ket{0^{\otimes n}}\bigr|^2.
\end{equation}
Assuming anticoncentration along with average-case hardness (which exploits that $\PIQP(s)= 4^{-n} |\mZ(s)|^2$, where $\mZ(s)$ is a suitable complex-parameter Ising partition function, believed $\#$P-hard to approximate), it is classically intractable to approximately sample from $\PIQP$ within constant total-variation (TV) accuracy unless the polynomial hierarchy ($\mathrm{PH}$) collapses~\cite{Bremner:2016aa,Boixo2018,Hangleiter2018}. Recall that the TV distance between two
probability distributions $P$ and $Q$ is $D(P,Q)
\equiv \frac{1}{2}\sum_{s} |P(s)-Q(s)|$.

Two key differences distinguish DADQC from IQP:
(i) IQP has no transverse field. To overcome this, we will exactly factor out the $\HX$ dynamics by appropriate one-qubit rotations implemented via $\UL$, and concentrate almost all of the integrated $B(t)$ weight into a short final window where $A(t)$ is already small; the resulting interaction-picture Hamiltonian is then nearly diagonal in the computational ($Z$) basis.
(ii) The IQP interaction graph $\mG$ need not match the hardware graph $\mH$, which often has bounded degree. We therefore introduce and analyze an implementable random-graph ensemble supported on $d$-regular subgraphs of $\mH$ (see \cref{def:QPU-R-E} below), aiming to match hardware graph constraints.

With these ingredients, we choose $\{h_i,J_{ij}\}$ so that the remaining diagonal gate reproduces the appropriate IQP angles (first the $\pi/8$ grid, later a graph state) and show that the DADQC output distribution is within a constant TV distance of $\PIQP$. 
This, together with anticoncentration (which we prove for complete graphs) and average-case hardness, will suffice to demonstrate sampling hardness for DADQC.

\textit{Emulating IQP via DADQC}.---%
Let us first assume that $\mH=\mG$, so both the DADQC model of \cref{eq:full,eq:HTFI,eq:HX-HI} and the IQP circuit of \cref{eq:IQP-general} live on the same graph. We later consider the case where the interaction graph $\mG$ is a subgraph of the hardware $\mH$.

Let $\alpha(t)\equiv \int_0^t A(t') dt'$,
$\RX(\theta)\equiv e^{-i\theta \HX}$,
and $S\equiv \RX(\alpha(T)) = \bigotimes_{j\in\mV_{\mH}}e^{-i \a(T)X_j}$.
Transform the Ising Hamiltonian to the interaction picture with respect to the transverse field $A(t)\HX$:
$\tHI(t) \equiv  \RX(-\alpha(t)) \HI \RX(\alpha(t))$,
so that $\tHI(T)=S^\dagger \HI S$. 
Then
\beq
\label{eq:interaction}
U_A(T) = S \tU(T) ,\quad \tU(t)\equiv \mT e^{-i\int_0^t B(t') \tHI(t') dt'}.
\eeq
Fix $\delta>0$ and define the early/late time intervals $\mI_1=[0,T-\delta]$ and $\mI_2=[T-\delta,T]$, together with
\beq
\label{eq:early-late}
\eta \equiv \int_{\mI_1} B(t) dt \ , \quad \beta \equiv  \int_{\mI_2} B(t) dt \ , \quad \Delta\a\equiv  \int_{\mI_2} A(t) dt .
\eeq
We will choose schedules so that $\eta$ and $\Delta\alpha$ are small while $\beta$ is not.

Let $\norm{\cdot}$ denote the operator norm (largest singular value), and define $K \equiv \|[\HX,\HI]\|$.
We can show using a Duhamel-type bound~\cite{ReedSimonII} (see End Matter) that 
\begin{align}
\label{eq:Vapprox}
\left\| \tU(T)
- e^{-i\beta \tHI(T)} \right\|&\le
  \eta \norm{\HI} + K \beta \Delta\a .
\end{align}
In analogy to $\PIQP(s)$ [\cref{eq:IQP-P}], we denote by 
\beq
\label{eq:PUDAD}
\PDAD(s)
= \bigl| \braket{s|\UDAD|0^{\otimes n}} \bigr|^2 
\eeq 
the outcome distribution of the DADQC model.

We now choose the initial digital layer $\UL$ to ensure that the transverse field evolution is undone, and compare the resulting  DADQC evolution to an IQP-type target via:

\begin{mylemma}
\label{lem:main}
Let $\UL=S^\dagger W^{\otimes n}$, $H_Z=\beta \HI$, and let $U'_{\mathrm{R}}$ be an arbitrary unitary (not necessarily a product). Define
\beq
\hspace{-.15cm} \UpDAD \equiv U'_{\mathrm{R}} U_{\mathrm{A}}(T) \UL, 
\quad
\UpIQP \equiv U'_{\mathrm{R}} e^{-i H_Z} W^{\otimes n} .
\eeq

Then the induced output distributions satisfy
\beq
\label{eq:TV}
D\big(\PpDAD, P_{\UpIQP}\big)
\le \eta\|\HI\|+K\beta\D\a.
\eeq
\end{mylemma}
The proof uses unitary invariance, \cref{eq:Vapprox}, and the data processing inequality (see End Matter).

To ensure that $\beta \HI$ matches $H_Z$ as required in \cref{lem:main}, we choose the schedule $B(t)$ [which sets $\beta$ via \cref{eq:early-late}] and the static coefficients $\{h_i,J_{ij}\}$ so that
\begin{equation}
\label{eq:angles}
\beta J_{ij}, \beta h_i \in \bigl\{\frac{k\pi}{8}:k=0,\dots,7\bigr\}.
\end{equation}
Exact target angles as in \cref{eq:angles} may not be achieved experimentally, but the hardness results for IQP circuits are robust to angle errors~\cite{Bremner:2016aa}. 
Note that for arbitrary hardware graphs with $n=|\mV|$ qubits and $m=|\mE|$ edges, and with $|h_i|\le h_{\max}$ and $|J_{ij}|\le J_{\max}$, we have
$\|\HI\|\le \sum_{i\in \mV} |h_i| + \sum_{(i,j)\in \mE} |J_{ij}| \le n h_{\max}+ m J_{\max}$, and $
K=\|[\HX,\HI]\|\le 2\sum_{i\in \mV}|h_i|+4\sum_{(i,j)\in \mE}|J_{ij}| \le 2 n h_{\max}+ 4 m J_{\max}$.
If the hardware graph $\mH$ has bounded maximum degree $d=O(1)$, then [with \cref{eq:angles}]
both $\|\HI\|,K=\Theta(n)$. 
Turning the resulting approximation of an IQP circuit by DADQC into a sampling hardness statement leads us to our first main result, which essentially states that unless  $\mathrm{PH}$ collapses, no classical polynomial-time algorithm can (weakly) simulate the corresponding DADQC device within a small total variation distance.

\begin{mytheorem}[Supremacy for DADQC]
\label{thm:sup}
Assume the standard IQP conjectures (anticoncentration together with average-case \#P-hardness of approximating output probabilities). Then there exists a constant $\varepsilon_\star>0$ such that, for any $0<\varepsilon<\varepsilon_\star$, one can choose smooth monotone schedules $A(t)$ (decreasing to $0$) and $B(t)$ (increasing from $0$), and single-qubit layers $\UL,\UR$, with the following property:
even in the presence of small parameter errors $|\Delta h'_i|,|\Delta J'_{ij}| = O\bigl(1/(n+m)\bigr)$ in the scaled angles $h'_i = \beta h_i$ and $J'_{ij} = \beta J_{ij}$, any classical probabilistic polynomial-time algorithm that samples from a distribution within total variation distance $<\varepsilon/2$ of the DADQC output distribution (i.e., of $\UpDAD$) would imply that $\mathrm{PH}$ collapses (to its third level).
\end{mytheorem}

\textit{Proof sketch}. Using \cref{lem:main} we choose smooth schedules $A(t),B(t)$ and one-qubit layers so that the DADQC output is within TV distance $\le \varepsilon/4$ of an ideal IQP circuit. Small static parameter errors in the scaled angles contribute at most another $\varepsilon/4$. Thus, the DADQC output remains within TV distance $\le \varepsilon/2$ of an IQP circuit output. Any polynomial-time classical sampler within TV distance $<\varepsilon/2$ of the DADQC output would then approximate the IQP output within TV distance $<\varepsilon$, contradicting IQP sampling hardness unless $\mathrm{PH}$ collapses. See End Matter for
the full proof.

\textit{Example}.---%
We give an example of smooth schedules that meet the requirements of \cref{thm:sup}.  
Consider
\bes
\label{eq:sigmoids}
\begin{align}
\label{eq:AandB}
A(t)&=A_0 \bigl[1-s(t)\bigr],\quad B(t)=B_0  s(t),\\
\label{eq:sigmoids2}
\tilde{s}(t)&=\frac{1}{2} (1+\tanh\frac{t-(T-\delta)}{\mu}) ,
\end{align}
\ees
with $\delta,\mu>0$. To enforce $s(0)=0$ and $s(T)=1$, define the normalized schedule $s(t)\equiv \frac{\tilde{s}(t)-s_0}{\kappa}$, where $s_0\equiv \tilde{s}(0)$ and $\kappa\equiv \tilde{s}(T)-s_0$. We find, after changing variables and direct integration (see End Matter for details)
\beq
\label{eq:eta-Delalpha-ineq}
\eta \le  \frac{\ln 2}{2}\frac{B_0\mu}{\kappa} , \quad \D\a \le \frac{\ln 2}{2}\frac{A_0\mu}{\kappa} .
\eeq
Note that these upper bounds are independent of $\delta$. We also find $\kappa\ge 1-2e^{-2\min\{T-\delta,\delta\}/\mu}$, so $\kappa$ is exponentially close to $1$ in the small-$\mu$ limit.

Inserting this into the TV bound of \cref{lem:main} yields
\beq
D\bigl(\PpDAD,P_{\UpIQP}\bigr)
 \le  \frac{\ln 2}{2} \frac{\mu}{\kappa}\Bigl(B_0\|\HI\| + K \beta A_0\Bigr).
\eeq
For bounded-degree hardware, where we already showed that $\|\HI\|,K=\Theta(n)$, it suffices to choose
\beq
\mu  =  \frac{\varepsilon\kappa}{2\ln 2\bigl(B_0\|\HI\|+K \beta A_0\bigr)}
 =  \Theta(1/n),
\eeq
which ensures
$\eta\|\HI\|  \le  \frac{\varepsilon}{4}$ and $K \beta \Delta\alpha  \le  \frac{\varepsilon}{4}$,
and thus $D\bigl(\PpDAD,P_{\UpIQP}\bigr)\le \varepsilon/2$ without imposing any constraint on $\delta$.  
The scaling $\mu=\Theta(1/n)$ means the transition window narrows with system size; this clearly cannot be sustained for arbitrarily large $n$, but a supremacy demonstration just needs to reach $n$ large enough to be beyond classical feasibility. Actual values are system-specific and are left for future work.

\textit{Supremacy under hardware graph constraints}.---%
Standard IQP hardness results often assume an all-to-all interaction graph~\cite{Bremner:2016aa}. In contrast, a variety of systems exhibit a fixed bounded-degree hardware graph $\mH = (\mV_{\mH},\mE_{\mH})\ne \mG$ (e.g., the Zephyr graph of quantum annealers~\cite{Boothby2020Zephyr}), so are not fully connected. Minor-embedding~\cite{Choi1,Choi2} or the LHZ scheme~\cite{Lechner:2015} can impose effective all-to-all connectivity, but this changes the graph ensemble in other ways~\cite{Albash:2016jk} for which anticoncentration is not guaranteed.

A natural alternative is to work with sparse random ensembles for which anticoncentration is available. In particular, for uniform random $d$-regular graphs (where every vertex has degree $d$ and each such graph is chosen with equal probability) combined with random $XY$-plane product measurements, recent work proves anticoncentration for constant $d\ge 3$~\cite{Ghosh2025PRXQ}. However, a uniformly random $d$-regular graph on $n$ labeled vertices will, with high probability, include edges outside $\mE_{\mH}$, so this ensemble cannot be realized by simply selecting native couplers on a fixed hardware graph.

Therefore, to make a hardware-faithful supremacy statement, we instead introduce a ``QPU-restricted'' ensemble which is natively implementable and matches our analysis. 

\begin{mydefinition}[QPU-restricted graph ensemble]
\label{def:QPU-R-E}
Let $\mH_n=(\mV_{\mH},\mE_{\mH})$ be a fixed simple $D$-regular hardware graph on $n$ labeled vertices, and let $3\le d\le D$ be a fixed constant ($n$-independent). 
The \emph{$d$-factors} of a graph $\mH$ are
\beq
\label{eq:dfactor}
\mF_d(\mH) \equiv \{\mG=(\mV,\mE_{\mG}) : \mE_{\mG}\subseteq \mE_{\mH},\deg_{\mG}(v)=d\ \forall v\in\mV\}.
\eeq
The \emph{QPU-restricted graph ensemble} is the probability space $\mathrm{Unif}[\mF_d(\mH_n)]$ in which, for each run of $\UDAD$, a graph $\mG\in\mF_d(\mH_n)$ is drawn uniformly at random.
\end{mydefinition}
We implicitly assume $\mF_d(\mH_n)\ne\emptyset$, which holds under mild parity conditions on $n$ and $d$ from standard results on 
$d$-factors of regular graphs~\cite{Tutte:1952aa,LovaszPlummer}.

To match the IQP measurement model used in anticoncentration proofs, we randomize the final measurement basis.
We choose i.i.d. angles $\{\theta_i\}_{i=1}^n$ uniformly from $[0,2\pi)$, 
and let 
\beq
\label{eq:RZ-UR}
R_Z(\theta)\equiv e^{-i\theta Z/2},\quad U_{\mathrm{R}}^{(\theta)}\equiv \bigotimes_{i=1}^n \Bigl(W R_Z(\theta_i)\Bigr).
\eeq
Instead of \cref{eq:IQP-general}, consider the IQP-type circuit
\bes
\label{eq:IQP-new}
\begin{align}
\UIQP^{(\theta)}
&=  U_{\mathrm{R}}^{(\theta)} e^{-i H'_Z}  W^{\otimes n} = W^{\otimes n} e^{-i H''_Z } W^{\otimes n}\\
\label{eq:HZ'}
H'_Z&=\sum_{i\in\mV_{\mG}} v_i Z_i + \frac{\pi}{4}\sum_{(i,j)\in \mE_{\mG}} Z_i Z_j,
\end{align}
\ees
with arbitrary but fixed single-qubit angles $\{v_i\}$, and where $H''_Z = H'_Z + \frac{1}{2}\sum_i \theta_i Z_i$. 
Thus, inserting $\UR^{(\theta)}$ 
is equivalent to shifting the one-qubit angles $v_i \mapsto v_i + \theta_i/2$.
This also amounts to randomizing the final measurement basis in the $XY$ plane, since measuring $Z_i$ after $\UR^{(\theta)}$ is equivalent to measuring the observable
$\UR^{(\theta)\dagger} Z_i \UR^{(\theta)} = R_Z(-\theta_i) X_i R_Z(\theta_i) = \cos\theta_i X_i - \sin\theta_i Y_i$,
i.e., an $XY$-plane measurement at angle $\theta_i$ from the $X$ axis. 
Setting the two-qubit phases in $H'_Z$ to $\pi/4$ yields the graph-state entangler used in the anticoncentration results of Ref.~\cite{Ghosh2025PRXQ}.

Next, we introduce a hardness conjecture tailored to simple $D$-regular graphs, as representatives of bounded-degree hardware graphs.

\begin{myconjecture}[QPU-restricted average-case Ising hardness]
\label{conj:host-avgcase}
Let $\{\mH_n\}$ be a family of hardware graphs as in \cref{def:QPU-R-E}.  
For each $n$, sample $\mG \sim \mathrm{Unif}[\mF_d(\mH_n)]$ and 
$\theta \sim \mathrm{Unif}([0,2\pi)^n)$ independently.  
Let $H''_Z(\mG,\theta)$ denote the $Z$-diagonal Hamiltonian in \cref{eq:IQP-new} (so the
single-qubit angles are shifted by $\theta_i/2$), and let $\mZ(\mG,\theta)$ be its
complex-temperature Ising partition function.  
Then there exists a constant $\delta_1>0$ such that for every algorithm 
$A \in \mathrm{FBPP}^{\mathrm{NP}}$ (bounded-error randomized polynomial time with access to an NP oracle), if $A$ approximates $\mZ(\mG,\theta)$ within multiplicative error 
$1 \pm 1/\mathrm{poly}(n)$ on a fraction $>\delta_1$ of instances $(\mG,\theta)$ drawn as above, 
then $\mathrm{PH}$ collapses.
\end{myconjecture}

This is the natural analogue of the standard IQP average-case conjecture~\cite{Bremner:2016aa}.
Next, we introduce an anticoncentration conjecture applicable to the QPU-restricted ensemble.

\begin{myconjecture}[Anticoncentration for QPU-restricted IQP ensemble]
\label{conj:host-anticonc}
Let $\{\mH_n\}$ be a family of hardware graphs as in \cref{def:QPU-R-E}.
For each $n$, draw $\mG\sim \mathrm{Unif}[\mF_d(\mH_n)]$
and $\theta\sim \mathrm{Unif}[0,2\pi)^n$. 
For any fixed choice of single-qubit angles $\{v_i\}$ in \cref{eq:HZ'}, there exist constants $a,b>0$ (depending only on $d$)
such that for every $s\in\{0,1\}^n$,
\beq
\label{eq:anticonc}
\Pr_{\mG,\theta} \bigl[ P_{\UIQP^{(\theta)}}(s)\ge a 2^{-n} \bigr] \ge b.
\eeq
\end{myconjecture}

\cref{conj:host-anticonc} applies directly, e.g., to the Zephyr graph (ignoring boundary effects)~\cite{Boothby2020Zephyr}, for any fixed $d\ge 3$. This connects our results directly to current quantum annealing hardware graphs. It also applies to superconducting-cavity architectures for digital-analog simulation~\cite{Lamata:2017aa}.

There is a special case in which---by appealing to Ref.~\cite{Ghosh2025PRXQ}---we can promote \cref{conj:host-anticonc} to a result, where we allow $
D$ to scale with $n$.

\begin{myproposition}[Anticoncentration for complete hardware graphs]
\label{prop:Kn-anticonc}
Let $\mH_n=K_n$ be the complete graph on $n$ labeled vertices and fix $d\ge 3$.  
Then the conclusion of \cref{conj:host-anticonc} holds for this choice of hardware graph family. 
\end{myproposition}

The proof is given in End Matter.
The all-to-all case in \cref{prop:Kn-anticonc} is directly relevant for hardware modalities with dense connectivity graphs, such as trapped ions or neutral atoms.  
In such systems, programmable Ising interactions on arbitrary pairs of ions~\cite{linkeExperimentalComparisonTwo2017,Monroe:2021aa,Kim:2010aa,Islam:2013mi,DeCross:2025aa} or atoms~\cite{Glaetzle:2017aa,Bernien:2017aa,Scholl:2021aa,Ebadi:22} are routinely implemented, so the QPU-restricted graph ensemble for $\mH_n=K_n$ can be realized natively.  

We are finally ready for the sampling hardness result, which formalizes quantum supremacy for DADQC using $D$-regular hardware graphs.

\begin{mytheorem}[Supremacy for DADQC on a fixed $D$-regular hardware graph family]
\label{thm:host-sup}
Let $\{\mH_n\}$ be a family of hardware graphs as in \cref{def:QPU-R-E}.
For each $n$ and run of $\UDAD$:
(1) Form an ensemble $(\mG,\t)$ by drawing $\mG\sim\mathrm{Unif}(\mF_d(\mH_n))$
and $\theta\sim \mathrm{Unif}[0,2\pi)^n$. 
(2) Set $\UL=S^\dagger W^{\otimes n}$ and $U_{\mathrm{R}}=U_{\mathrm{R}}^{(\theta)}$ as in \cref{eq:RZ-UR}.
(3) Choose smooth, monotone schedules $A(t),B(t)$ as in \cref{eq:HTFI} and pick $\beta=\int_{\mI_2} B(t) dt$ so that the analog block implements $e^{-i\beta \HI}$ with $\beta J_{ij}=\pi/4$ for $(i,j)\in \mE_{\mG}$ and $\beta J_{ij}=0$ otherwise.
Then, 
assuming Conjectures~\ref{conj:host-avgcase} and \ref{conj:host-anticonc}, there exist constants $\varepsilon_\star,\delta_2>0$ such that for any $0<\varepsilon<\varepsilon_\star$ one can choose schedules $A(t),B(t)$ so that $\eta\|\HI\|+K \beta \Delta\alpha\le \varepsilon/2$. For such schedules, any classical algorithm that, on a fraction $> \delta_2$ of instances from the ensemble $(\mG,\theta)$ samples within TV distance $<\varepsilon/2$ of the DADQC output distribution would imply that $\mathrm{PH}$ collapses down to its third level. Thus (unless 
$\mathrm{PH}$ collapses) no such classical algorithm exists.
\end{mytheorem}

\noindent\emph{Proof sketch.} By \cref{lem:main} applied with $U_{\mathrm{R}}=U_{\mathrm{R}}^{(\theta)}$, the device output is within $\varepsilon/2$ of $P_{\UIQP^{(\theta)}}$. \Cref{conj:host-avgcase} 
rules out an $\mathrm{FBPP}^{\mathrm{NP}}$
 algorithm that multiplicatively approximates 
$\mZ(\mG,\theta)$ on more than a $\delta_1$-fraction of instances (unless PH collapses), and the ensemble $(\mG,\theta)$ anticoncentrates by \cref{conj:host-anticonc} or \cref{prop:Kn-anticonc}. A standard Stockmeyer-type argument then yields the claimed PH collapse if a classical sampler can achieve TV distance $<\varepsilon$ on $> \delta_2$ of the instances. See End Matter for the full proof.

We stress that \cref{thm:host-sup} relies only on \cref{lem:main} 
and the choice $\beta J_{ij}=\pi/4$ on active edges together with random $XY$-plane product measurements; 
the $\pi/8$-grid used in \cref{thm:sup} is not required here. In particular, the $\{\beta h_i\}$ can be arbitrary, e.g., set to zero for convenience; any static $Z$-field offsets simply shift the random measurement angles and therefore do not affect the argument.

\textit{Conclusion}.---%
We analyzed a minimal digital-analog-digital quantum computing  (DADQC) model: a single continuous transverse-field Ising block sandwiched between arbitrary one-qubit layers. 
We proved (\cref{lem:main}) that with properly designed, smooth schedules compatible with present-day annealers, the device unitary is within constant TV distance of an IQP circuit with angles on the $\{\frac{\pi}{8}k\}$ grid. This leads to a supremacy result for the DADQC model when the hardware graph matches the IQP graph (\cref{thm:sup}).

To obtain a hardware-faithful result on fixed-topology devices, we matched the IQP ensemble to what such hardware can implement: on each run, the device implements a uniformly random $d$-factor of the fixed $D$-regular hardware graph (\cref{def:QPU-R-E}) and applies random $XY$-plane product measurements via the final one-qubit layer $U_{\mathrm{R}}^{(\theta)}$. 
Combining our DADQC-IQP closeness bound with a natural average-case hardness conjecture for the corresponding complex-temperature Ising partition functions (\cref{conj:host-avgcase}), and an anticoncentration conjecture for $D$-regular hardware graphs (\cref{conj:host-anticonc}) that we prove for complete graphs (\cref{prop:Kn-anticonc}), we obtain \cref{thm:host-sup}: the DADQC model can achieve quantum supremacy on fixed $D$-regular hardware graphs.

Conceptually, this places DADQC alongside RCS, boson sampling, and QAOA as a viable route to quantum supremacy, while remaining directly compatible with analog hardware. 
Our results motivate several directions: (i) proving Conjectures~\ref{conj:host-avgcase} and \ref{conj:host-anticonc}; (ii) quantifying robustness under time-dependent control errors and decoherence; and (iii) implementing the DADQC model on current QPUs, such as quantum annealers, trapped ions, and neutral atoms.

\textit{Acknowledgments}.---Discussions with the D-Wave team are gratefully acknowledged. This material is based upon work supported by, or in part by, the U. S. Army Research Laboratory and the U. S. Army Research Office under contract/grant number W911NF2310255.

\bibliographystyle{apsrev4-2}

%

\begin{center}
    \textbf{End Matter}
\end{center}

\textit{Related work}.---%
Our work is related in spirit to ``digital-analog quantum computation'' (DAQC)~\cite{Parra-Rodriguez:2020aa}, where homogeneous Ising interactions serve as an analog resource that, when combined with single-qubit rotations, enable universality and efficient synthesis of inhomogeneous multi-body dynamics. In contrast, we focus on a minimal digital-analog-digital construction aimed at a complexity-theoretic sampling result for sparse, hardware-faithful architectures rather than on universality and compilation.

Also related is an abstract scheme that adiabatically prepares specially engineered Feynman-Kitaev history-state (clock) ground states~\cite{Kitaev:book,aharonov_adiabatic_2007}) and then measures them to achieve IQP-type quantum supremacy~\cite{fujii2018quantum}. In contrast, our diabatic scheme is implementable in real devices; it avoids clock registers and the non-native 
$k$-local interactions that typically require gadget reductions and additional embedding overhead~\cite{kempe:1070}.

\textit{Derivation of \cref{eq:Vapprox}}.---%
We write $\tilde \HI(t)=F(\alpha(t))$, where $F(\theta)\equiv \RX(-\theta) \HI \RX(\theta)$. Then
\beq
\label{eq:F'}
F'(\theta) = i[\HX,F(\theta)] 
= i \RX(-\t) [\HX,\HI] \RX(\t) .
\eeq
The fundamental theorem of calculus gives
\beq
\Delta\tilde \HI \equiv \tilde \HI(T)-\tilde \HI(t)
=\int_{\alpha(t)}^{\alpha(T)} F'(\theta) d\theta .
\eeq
Substituting \cref{eq:F'}, taking operator norms, and using unitary invariance,
\beq
\|\Delta\tilde \HI\|
\le \int_{\alpha(t)}^{\alpha(T)} \|[\HX,\HI]\| d\theta
= K |\alpha(T)-\alpha(t)|,
\eeq

In our setting $A(\cdot)\ge 0$, so $\alpha(T)\ge\alpha(t)$, which yields, using \cref{eq:early-late}:
\beq
\label{eq:HIt-T} 
\|\tilde \HI(T)-\tilde \HI(t) \| \le K \D\a \ , \quad t\in\mI_2
\eeq
where we also used that for $t\in\mI_2=[T-\delta,T]$ one has $[t,T]\subseteq\mI_2$.

Next, we need the following standard result (Duhamel's inequality)~\cite{ReedSimonII}:
for bounded Hamiltonians $H_1(t)$ and $H_2(t)$ and $t\in[0,T]$, let $U_k(t,s)$, $k\in\{1,2\}$, denote the unitary propagators solving
$i \partial_t U_k(t,s)=H_k(t) U_k(t,s)$ with the initial condition $U_k(s,s)=I$.
Then, for all $0\le s\le t\le T$, $\bigl\|U_1(t,s)-U_2(t,s)\bigr\|  \le  \int_s^t \bigl\|H_1(\tau)-H_2(\tau)\bigr\| d\tau$.
In particular, writing $U_k(t)\equiv U_k(t,0)$,
\beq
\label{eq:Duhamel}
\bigl\|U_1(T)-U_2(T)\bigr\|  \le  \int_0^T \bigl\|H_1(t)-H_2(t)\bigr\| dt.
\eeq
Now split the integral at $T-\delta$
and write $\tU(T)=\tU_2\tU_1$ with 
$\tU_1 = \mT e^{-i \int_{\mI_1} B(t)\tHI(t)dt}$ 
and 
$\tU_2 = \mT e^{-i \int_{\mI_2} B(t)\tHI(t)dt}$. Then
$\tU(T)-e^{-i\beta\tHI(T)}
=(\tU_2-e^{-i\beta\tHI(T)})\tU_1 + e^{-i\beta\tHI(T)}(\tU_1-I)$.
Taking norms, using the triangle inequality and $\|U\|=1$ for unitaries, we obtain
\begin{align}
\label{eq:split}
\|\tU(T)-e^{-i\beta\tHI(T)}\|
&\le \|\tU_2-e^{-i\beta\tHI(T)}\|
   + \|\tU_1-I\| .
\end{align}
For the first interval, using \cref{eq:Duhamel}, $B(t)\ge 0$, and unitary invariance of the operator norm:
\beq
\label{eq:early-bound}
\left\|\tU_1 - I\right\| \le \int_{\mI_1} B(t)\norm{\tHI(t)} dt \le
\eta\norm{\HI} .
\eeq
Using \cref{eq:Duhamel} again:
\bes
\label{eq:late-bound}
\begin{align}
\left\|\tU_2 - e^{-i\beta\tHI(T)}\right\| &\le \int_{\mI_2} B(t)\norm{\tHI(t)-\tHI(T)} dt \\
  &\le \beta K\Delta\a ,
\end{align}
\ees
where in the last inequality we used \cref{eq:HIt-T}. Combining \cref{eq:split,eq:early-bound,eq:late-bound} yields \cref{eq:Vapprox}.\\

\textit{Proof of \cref{lem:main}}.---%
Using \cref{eq:full,eq:interaction} and $\UL=S^\dagger W^{\otimes n}$ we have
\beq
\UpDAD = U_{\mathrm{R}} U_{\mathrm{A}}(T) \UL
= U_{\mathrm{R}} S\tU(T) S^\dagger W^{\otimes n}.
\eeq
Thus, with $\UpIQP=U_{\mathrm{R}} e^{-i H_Z} W^{\otimes n}$ and $H_Z=\beta \HI$,
\bes
\label{eq:U'DAD-Uid}
\begin{align}
&\|\UpDAD - \UpIQP\|
= \big\| U_{\mathrm{R}} \big(S\tU(T)S^\dagger - e^{-i\beta\HI}\big) W^{\otimes n} \big\| \notag \\
\label{eq:U'DAD-Uid-2}
&\quad = \|S\tU(T)S^\dagger - e^{-i\beta\HI}\| \\
\label{eq:U'DAD-Uid-3}
&\quad = \|\tU(T) - e^{-i\beta\tHI(T)}\|
\le \eta\|\HI\|+K\beta\D\a,
\end{align}
\ees
where we used unitary invariance of the operator norm, and then \cref{eq:Vapprox} together with $\tHI(T)=S^\dagger \HI S$. For pure states $\ket{\psi_k}=U_k\ket{0}^{\otimes n}$, the trace-norm distance satisfies $D_{\mathrm{tr}}(\psi_1,\psi_2)\le \|\ket{\psi_1}-\ket{\psi_2}\| \le \|U_1-U_2\|$~\cite{Fuchs:99}, and measurement cannot increase  $D_{\mathrm{tr}}$ (data processing inequality)~\cite{Watrous:book}:
\beq
\label{eq:data-proc}
D(P_{U_1},P_{U_2}) \le D_{\mathrm{tr}}(\psi_1,\psi_2)\le \|U_1-U_2\|.
\eeq
Applying this to $U_1=\UpDAD$ and $U_2=\UpIQP$ yields \cref{eq:TV}. This concludes the proof.\\

\textit{Proof of \cref{thm:sup}}.---%
Fix $0<\varepsilon<\varepsilon_\star$, where $\varepsilon_\star$ is a constant below the anticoncentration-dependent threshold appearing in the constant-error IQP reductions (see, e.g.,~\cite{Bremner:2016aa,Hangleiter2018}).  
Choose smooth, monotone schedules $A(t)$ and $B(t)$ so that
$\eta\|\HI\|+K \beta \Delta\alpha  \le  \varepsilon/4$; 
this is achievable as shown by the example below.
Set $\UL=S^\dagger W^{\otimes n}$ and $\UR=W^{\otimes n}$; choose parameter values so that \cref{eq:angles} holds. Then $\UpIQP=\UR e^{-iH_Z}W^{\otimes n}$ coincides with $\UIQP$ and $\UpDAD=\UDAD$.
By \cref{lem:main},
\begin{equation}
\label{eq:TV-key}
D \bigl(\PDAD,  \PIQP \bigr) \le \eta\|\HI\|+K \beta \Delta\alpha \le \varepsilon/4 .
\end{equation}
Now allow static miscalibration of the scaled angles, $h'_i\mapsto h'_i+\Delta h'_i$ and $J'_{ij}\mapsto J'_{ij}+\Delta J'_{ij}$. As shown in End Matter, if 
$|\Delta h'_i|,|\Delta J'_{ij}|\le \varepsilon/\bigl(4(n+m)\bigr)$ then the resulting perturbation $\Delta H'_I$ of the diagonal Hamiltonian satisfies $\|e^{-i(H'_I+\Delta H'_I)}-e^{-iH'_I}\|\le\varepsilon/4$. This contributes at most $\varepsilon/4$ to the operator-norm bound used in the proof of \cref{lem:main}, and hence to the TV bound in \cref{eq:TV-key}. Thus
$D(\PDAD,\PIQP)\le \varepsilon/2$
for the implemented instance.
Now suppose there exists a polynomial-time classical sampler whose output distribution $P_{\mathrm{cl}}$ obeys $D (P_{\mathrm{cl}},\PDAD)<\varepsilon/2$ on the same instance. By the triangle inequality,
\beq
D\bigl(P_{\mathrm{cl}},\PIQP\bigr)
\le
D\bigl(P_{\mathrm{cl}},\PDAD\bigr)
+ D\bigl(\PDAD,\PIQP\bigr)
< \varepsilon .
\eeq
Since $\varepsilon<\varepsilon_\star$, the constant-error IQP sampling hardness then implies a collapse of $\mathrm{PH}$, a contradiction. \\

\textit{Effect of parameter errors; derivation of \cref{eq:TV-key}}.---%
Suppose the (scaled) target parameters $h'_i=\b h_i$ and $J'_{ij} = \b J_{ij}$ in
$\HI'\equiv \b\HI=\sum_{i\in\mV} h'_i Z_i + \sum_{(i,j)\in\mE} J'_{ij} Z_i Z_j$
are implemented with errors
$|\Delta h'_i|, |\Delta J'_{ij}|\le \delta_{\mathrm{par}}$ (here $\mV$ and $\mE$ denote the vertex and edge set of the relevant graph, e.g., $\mH_n$ or $\mG$).
Let 
$\b\D\HI = \D \HI'  =  \sum_{i\in\mV} \Delta h'_i  Z_i  +  \sum_{(i,j)\in\mE} \Delta J'_{ij}  Z_i Z_j$
be the resulting perturbation, 
and let $n=|\mV|$ and $m=|\mE|$. Then
\beq
\|\Delta \HI'\|\le \sum_{i\in\mc V}|\Delta h'_i| + \sum_{(i,j)\in\mc E}|\Delta J'_{ij}|
\le (n+m) \delta_{\mathrm{par}}.
\eeq
If the maximum degree is $d$,
then $m\le d n/2$ and the bound can be written as
$\|\Delta \HI'\|\le \bigl(1+d/2\bigr)n\delta_{\mathrm{par}}$.
Moreover, the difference between the intended and implemented unitaries is
\beq
\begin{aligned}
&\|e^{-i(\HI'+\D \HI' )}-e^{-i\HI'}\|
=\| \int_0^1 \frac{d}{ds}  e^{-i(\HI'+s\D \HI' )}  ds \|\\
&\ \ =\| \int_0^1 \bigl(-i e^{-i(\HI'+s\D \HI' )} \D \HI' \bigr)  ds \|
\le \|\D \HI' \| .
\end{aligned}
\eeq
Thus, parameter miscalibration contributes at most $\bigl(1+d/2\bigr)n\delta_{\mathrm{par}}$ additively to the operator-norm bound in \cref{eq:U'DAD-Uid-2}, and hence also to the TV bound in \cref{lem:main}.
Choosing, e.g., $\delta_{\mathrm{par}}\le \varepsilon/\bigl(4( 1+d/2)n\bigr)$ ensures that the added contribution is $\le \varepsilon/4$.

With $|h_i|\le h_{\max}$ and $|J_{ij}|\le J_{\max}$ on a graph with $n$ vertices
and $m$ edges, $\|\HI\| \le n h_{\max}+ m J_{\max}$ and $K=\|[\HX,\HI]\| \le 2n h_{\max}+ 4 m J_{\max}$.
Hence the TV bound $\eta\|\HI\| + K\beta\Delta\alpha$ scales with $n+m$.
For bounded degree $d=O(1)$, $m=O(n)$ and one can take $\mu=\Theta(1/n)$; for
dense graphs (the subject of \cref{prop:Kn-anticonc}) we have $m=\Theta(n^2)$, so $\mu=\Theta(1/n^2)$, which is more experimentally demanding.\\

\textit{Derivation of \cref{eq:eta-Delalpha-ineq}}.---%
On the early interval $\mI_1=[0,T-\delta]$ we have $0\le s_0\le \tilde{s}(t)\le \frac12$ and $\kappa\le 1$, so
\beq
s(t)=\frac{\tilde{s}(t)-s_0}{\kappa} \le \frac{\tilde{s}(t)}{\kappa}.
\eeq
Hence
\beq
\int_0^{T-\delta} B(t) dt
= B_0 \int_0^{T-\delta} s(t) dt
\le \frac{B_0}{\kappa} \int_0^{T-\delta} \tilde{s}(t) dt.
\eeq
A direct integration of the unnormalized sigmoid $\tilde{s}(t)$ gives, with $x\equiv (T-\delta)/\mu$,
\beq
\int_0^{T-\delta} \tilde{s}(t) dt
= \frac{\mu}{2}\left[x-\ln\cosh x\right],
\eeq
so
\beq
\eta \equiv \int_0^{T-\delta} B(t) dt
\le \frac{B_0\mu}{2\kappa} \left[x-\ln\cosh x\right].
\eeq
Using $\cosh u\ge e^{u}/2$ for $u\ge 0$ implies $x-\ln\cosh x\le \ln 2$, hence
\beq
\eta \le \frac{\ln 2}{2} \frac{B_0\mu}{\kappa}.
\eeq

For the late interval $\mI_2=[T-\delta,T]$ we similarly have
\beq
\Delta\alpha \equiv \int_{T-\delta}^{T} A(t) dt
= A_0 \int_{T-\delta}^{T} [1-s(t)] dt.
\eeq
Using $s(t)=(\tilde{s}(t)-s_0)/\kappa$ we find
\beq
1-s(t) = \frac{\kappa+s_0-\tilde{s}(t)}{\kappa},
\eeq
so
\beq
\Delta\alpha
= \frac{A_0}{\kappa}\left[(\kappa+s_0)\delta - \int_{T-\delta}^{T} \tilde{s}(t) dt\right].
\eeq
A direct integration of $\tilde{s}(t)$ on $\mI_2$ yields, with $y\equiv \delta/\mu$,
\beq
\int_{T-\delta}^{T} \tilde{s}(t) dt
= \frac{\delta}{2} + \frac{\mu}{2}\ln\cosh y.
\eeq
Noting that $\kappa+s_0 = \frac12(1+\tanh y)$ [since $\tilde{s}(T)=\frac12(1+\tanh y)$], we obtain
\beq
\Delta\alpha
= \frac{A_0\mu}{2\kappa}\left[y\tanh y - \ln\cosh y\right].
\eeq
Define $f(y)\equiv y\tanh y - \ln\cosh y$ for $y\ge 0$. Then $f(0)=0$, $f'(y)=y \sech^2 y\ge 0$, so $f$ is increasing, and
$\lim_{y\to\infty} f(y) = \ln 2$.
Thus $0\le f(y)\le \ln 2$ for all $y\ge 0$, and we obtain
\beq
\Delta\alpha \le \frac{\ln 2}{2} \frac{A_0\mu}{\kappa}.
\eeq

Moreover, using $x\equiv \frac{T-\delta}{\mu}$ and $y\equiv \frac{\delta}{\mu}$ again:
\beq
\kappa \equiv \tilde{s}(T)-\tilde{s}(0)
= \frac{1}{2}\bigl[\tanh x + \tanh y\bigr],
\eeq
so for $x,y\gg 1$,
\bes
\begin{align}
 \kappa &= 1 - e^{-2x} - e^{-2y} + O\bigl(e^{-4\min\{x,y\}}\bigr)
\\
&\ge 1 - 2e^{-2\min\{T-\delta,\delta\}/\mu},
\end{align}
\ees
as claimed in the main text.\\

\textit{Motivation for $d\ge 3$ in \cref{def:QPU-R-E}}.---%
The condition $d\ge 3$ is needed since otherwise graph-state sampling with product measurements (which we use below) is efficiently simulable: $2$-regular graphs have treewidth $2$, independent of $n$; graph states on any graph with bounded treewidth can be efficiently classically simulated via tensor-network contraction~\cite[Theorem~1.6]{MarkovShi2008} (see also Refs.~\cite{Van-den-Nest:2007aa,Ghosh:2023aa}).\\

\textit{Motivation for choosing $\pi/4$ in \cref{eq:HZ'}}.---%
Note that 
\begin{align}
\prod_{(i,j)\in \mE_{\mG}} e^{-i\frac{\pi}{4} Z_i Z_j}
&= e^{i\frac{\pi}{4}|\mE_{\mG}|} \Bigl(\prod_{i} e^{-i\frac{\pi}{4}\deg_{\mG}(i) Z_i}\Bigr)\notag \\
&\quad\times 
\prod_{(i,j)\in \mE_{\mG}} \mathrm{CZ}_{ij},
\end{align}
where the prefactor $e^{i\frac{\pi}{4}|\mE_{\mG}|}$ is a global phase, and
\beq
\mathrm{CZ}_{ij}
= e^{-i\pi/4}  e^{i\frac{\pi}{4} Z_i}  e^{i\frac{\pi}{4} Z_j}  e^{-i\frac{\pi}{4} Z_i Z_j}.
\eeq
I.e., setting $\beta J_{ij}=\pi/4$ implements the graph-state entangler $\prod\mathrm{CZ}_{ij}$ modulo one-qubit $Z$ rotations, which we can absorb into $H'_Z$ [shift $v_i\mapsto v_i-\frac{\pi}{4}\deg_{\mG}(i)$]. Recall that a graph state is prepared by applying $\mathrm{CZ}_{ij}$ on all edges to $\ket{+}^{\otimes n}$~\cite{Hein:2004aa}.\\

\textit{Proof of \cref{prop:Kn-anticonc}}.---%
Let $\{\mH_n\}$ be a family of hardware graphs as in \cref{def:QPU-R-E}.
Draw $\mG\sim \mathrm{Unif}[\mF_d(\mH_n)]$ and $\theta\sim \mathrm{Unif}[0,2\pi)^n$.
For each such graph $\mG$ and angles $\t$, let $p_{\mG,\theta}(s)$ denote the output bitstring distribution of the corresponding graph state, i.e.,
for any bitstring $s\in\{0,1\}^n$
 \beq
 p_{\mG,\theta}(s) = \PIQP^{(\theta)}(s)
 = \bigl|\bra{s}\UIQP^{(\theta)}\ket{0^{\otimes n}}\bigr|^2 .
 \eeq
 Recall that 
 $\UIQP^{(\theta)}=  U_{\mathrm{R}}^{(\theta)} e^{-i H'_Z}  W^{\otimes n}$, where $U_{\mathrm{R}}^{(\theta)} = \bigotimes_{i=1}^n \Bigl(W R_Z(\theta_i)\Bigr)$ [\cref{eq:RZ-UR}], and where $H'_Z=\sum_{i\in\mV_\mG} v_i Z_i + \frac{\pi}{4}\sum_{(i,j)\in \mE_{\mG}} Z_i Z_j$,
 with arbitrary but fixed single-qubit angles $\{v_i\}$ [\cref{eq:IQP-new}].

 We proceed in two main steps, one for the first moment of $p_{\mG,\theta}(s)$, the other for the second moment, which we specialize to complete graphs.

 \begin{mylemma}
 \label{lem:first-moment}
 The first moment satisfies
 \begin{equation}
 \label{eq:bip-first-moment}
 \mathbb{E}_{\mG,\theta}[p_{\mG,\theta}(s)] = 2^{-n}.
 \end{equation}
 \end{mylemma}

 \begin{proof}
 The informal reason is that this follows by symmetry. For each fixed $\mG$ and $\t$, $\sum_s p_{\mG,\t}(s)=1$.
The ensemble over $(\mG,\t)$ is symmetric over output strings $s$; hence $\mathbb{E}_{\mG,\theta}[p_{\mG,\theta}(s)]$ is the same for all $s$, and summing over $s$ gives 
$2^n \mathbb{E}_{\mG,\theta}[p_{\mG,\theta}(s)]=1$.

A more complete argument that makes this explicit is the following.
 Writing $U \equiv e^{-iH'_Z} W^{\otimes n}$, we have $\UIQP^{(\theta)}=U_{\mathrm{R}}^{(\theta)}U$. Then
 \beq
 \label{eq:X}
 p_{\mG,\theta}(s)  =  \bra{0^{\otimes n}} U^\dagger \Bigl(U_{\mathrm{R}}^{(\theta)\dagger}\ketbra{s}U_{\mathrm{R}}^{(\theta)}\Bigr) U \ket{0^{\otimes n}}.
 \eeq
 Since the $\theta_i$ are i.i.d. uniform on $[0,2\pi)$ and $U_{\mathrm{R}}^{(\theta)}$ factorizes,
 \begin{align}
 &\mathbb{E}_{\theta} \left[ U_{\mathrm{R}}^{(\theta)\dagger}\ketbra{s}U_{\mathrm{R}}^{(\theta)} \right]\notag \\
 &\quad = \bigotimes_{i=1}^n \mathbb{E}_{\theta_i} \left[ R_Z(-\theta_i) \Bigl(W \ketbra{s_i} W\Bigr) R_Z(\theta_i)\right].
 \end{align}

 Using 
 \beq
 \label{eq:WsiW}
 W\ketbra{s_i}W = \frac{1}{2}\bigl(I + (-1)^{s_i} X\bigr)
 \eeq 
 we obtain
 \beq
 \label{eq:RXR}
 R_Z(-\theta_i) X R_Z(\theta_i)=\cos\theta_i X-\sin\theta_i Y .
 \eeq 
 Since $\mathbb{E}_{\theta_i}[\cos\theta_i] = \mathbb{E}_{\theta_i}[\sin\theta_i]=0$, we find
 \beq
 \mathbb{E}_{\theta_i} \left[ R_Z(-\theta_i) \Bigl(W \ketbra{s_i} W\Bigr) R_Z(\theta_i)\right]  =  \frac12 I.
 \eeq
 Therefore, \cref{eq:X} yields
 \beq
 \mathbb{E}_{\mG,\theta}[p_{\mG,\theta}(s)]
 = \bra{0^{\otimes n}} U^\dagger \Bigl(\frac12 I\Bigr)^{ \otimes n} U \ket{0^{\otimes n}}
 = 2^{-n},
 \eeq
 as claimed.
 \end{proof}

\begin{mylemma}
 \label{lem:second-moment}
Restrict to the case $\mH_n = K_n$, so that $\mG$ is drawn from 
$\mathrm{Unif}[\mF_d(K_n)]$, which coincides with the usual uniform 
distribution $\mG_r(d,n)$ on labeled $d$-regular graphs on $[n]$.
 Then the second moment satisfies
\begin{equation}
 \label{eq:bip-second-moment}
\mathbb{E}_{\mG,\theta}\bigl[p_{\mG,\theta}(s)^2\bigr]
\le\ C(d) 2^{-2n}.
\end{equation}
for some $n$-independent constant $C(d)$.
\end{mylemma}

\begin{proof}
A $d$-factor of $K_n$ is exactly a labeled $d$-regular graph on $n$ vertices, and conversely any labeled $d$-regular graph on $n$ vertices is a $d$-factor of $K_n$. Thus the uniform distribution on $\mF_d(K_n)$ coincides with the uniform distribution $\mG_r(d,n)$ on simple $d$-regular graphs on the vertex set $[n]$, used in Ref.~\cite[Cor.~1]{Ghosh2025PRXQ}. Hence we may write $\mathbb{E}_{\mG\sim \mG_r(d,n), \theta}$ and $\mathbb{E}_{\mG\sim\mathrm{Unif}(\mF_d(K_n)),\t}$ interchangeably; as above, we abbreviate the latter as $\mathbb{E}_{\mG,\t}$.

Define the normalized second moment
\begin{equation}
\label{eq:m2}
m_2(\mG,\theta) = 2^n \sum_{s\in\{0,1\}^n} p_{\mG,\theta}(s)^2.
\end{equation}
Ref.~\cite[Cor.~1]{Ghosh2025PRXQ} shows that, for $3\le d=o(n^{1/2})$ (since we fix $d$ this holds for us),
\begin{equation}
\mathbb{E}_{\mG\sim \mG_r(d,n), \theta}\bigl[m_2(\mG,\theta)\bigr]\le\ C(d)
\end{equation}
for all sufficiently large $n$, where $C(d)$ depends only on $d$ and is independent of $n$. Again, by symmetry $\mathbb{E}_{\mG, \theta}\bigl[p_{\mG,\theta}(s)^2\bigr]$
is independent of $s$ (see directly below); denote this common value by $q_n(d)$. Then 
\beq
\label{eq:E[m2]}
\mathbb{E}_{\mG,\t}[m_2(\mG,\theta)] = 2^n \sum_{s\in\{0,1\}^n} q_n(d) = 4^n q_n(d).
\eeq
Solving for $q_n(d)$ yielding \cref{eq:bip-second-moment}.

To complete the proof, let us show that 
$\mathbb{E}_{\mG,\theta}\bigl[p_{\mG,\theta}(s)^2\bigr]$ is independent of $s$. 
Indeed, let $\rho\equiv\ketbra{0^{\otimes n}}$, $N_s(\theta)\equiv U_{\mathrm{R}}^{(\theta)\dagger}\ketbra{s}U_{\mathrm{R}}^{(\theta)}$, and 
$M_s(\mG,\theta)\equiv U^\dagger N_s(\theta)U$, with 
$U\equiv e^{-iH'_Z}W^{\otimes n}$. Then, using \cref{eq:X} we have $p_{\mG,\theta}(s) = \Tr \Bigl[\rho  M_s(\mG,\theta)\Bigr]$, so that
\beq
p_{\mG,\theta}(s)^2
= \Tr \Bigl[(\rho\otimes\rho)  M_s(\mG,\theta)\otimes M_s(\mG,\theta)\Bigr] .
\eeq
Letting $\mU\equiv U\otimes U$, we can write 
\beq
M_s(\mG,\theta)\otimes M_s(\mG,\theta)
= \mU^\dagger \bigl(N_s(\theta)\otimes N_s(\theta)\bigr) \mU,
\eeq
and therefore
\beq
\mathbb{E}_{\theta}\bigl[M_s(\mG,\theta)\otimes M_s(\mG,\theta)\bigr]
= \mU^\dagger \mathbb{E}_{\theta}\bigl[N_s(\theta)\otimes N_s(\theta)\bigr] \mU .
\eeq
Thus it suffices to show that $\mathbb{E}_{\theta}\bigl[N_s(\theta)\otimes N_s(\theta)\bigr]$ is independent of $s$.

Because $U_{\mathrm{R}}^{(\theta)}$ factorizes and the angles $\{\theta_i\}$ are i.i.d.\ uniform on $[0,2\pi)$,
\beq
N_s(\theta)\otimes N_s(\theta)
= \bigotimes_{i=1}^n \Bigl(N_{s_i,i}(\theta_i)\otimes N_{s_i,i}(\theta_i)\Bigr),
\eeq
where $s_i\in\{0,1\}$ is the $i$'th bit of $s$, and $N_{s_i,i}(\theta_i)
= R_Z(-\theta_i) W\ketbra{s_i}W R_Z(\theta_i)$. The expectation over $\theta$ factorizes:
\beq
\mathbb{E}_{\theta}\bigl[N_s(\theta)\otimes N_s(\theta)\bigr]
= \bigotimes_{i=1}^n \mathbb{E}_{\theta_i}\bigl[N_{s_i,i}(\theta_i)\otimes N_{s_i,i}(\theta_i)\bigr] ,
\eeq
so it suffices to prove that for each site $i$,
\beq
\mathbb{E}_{\theta_i}\bigl[N_{0,i}(\theta_i)\otimes N_{0,i}(\theta_i)\bigr]
= \mathbb{E}_{\theta_i}\bigl[N_{1,i}(\theta_i)\otimes N_{1,i}(\theta_i)\bigr].
\eeq
Using \cref{eq:WsiW,eq:RXR},
\beq
N_{s_i,i}(\theta_i)
= \frac12\Bigl[I + (-1)^{s_i}\bigl(\cos\theta_i X - \sin\theta_i Y\bigr)\Bigr],
\eeq
Hence, using $\mathbb{E}_{\theta_i}[\cos\theta_i] = \mathbb{E}_{\theta_i}[\sin\theta_i] = \mathbb{E}_{\theta_i}[\cos\theta_i\sin\theta_i] =0$ and $\mathbb{E}_{\theta_i}[\cos^2\theta_i]=\mathbb{E}_{\theta_i}[\sin^2\theta_i]=\frac12$, we obtain, after expanding the product and taking the expectation value
\beq
\mathbb{E}_{\theta_i}\bigl[N_{s_i,i}(\theta_i)\otimes N_{s_i,i}(\theta_i)\bigr]
= \frac14 I\otimes I + \frac18\bigl(X\otimes X + Y\otimes Y\bigr),
\eeq
which is manifestly independent of $s_i$. Thus, 
$\mathbb{E}_{\theta}\bigl[N_s(\theta)\otimes N_s(\theta)\bigr]$ is the same for all $s\in\{0,1\}^n$, and since conjugating by $U\otimes U$ does not introduce any $s$-dependence,
$\mathbb{E}_{\theta}\bigl[M_s(\mG,\theta)\otimes M_s(\mG,\theta)\bigr]$
is also independent of $s$. Finally,
\beq
\mathbb{E}_{\mG,\theta}\bigl[p_{\mG,\theta}(s)^2\bigr]
= \Tr\Bigl[(\rho\otimes\rho) \mathbb{E}_{\mG,\theta}\bigl[M_s(\mG,\theta)\otimes M_s(\mG,\theta)\bigr]\Bigr]
\eeq
is independent of $s$, as claimed.
\end{proof}

The Paley-Zygmund inequality~\cite{Vershynin2018HighDimProb} states that 
for any nonnegative random variable $x$ with $0<\mathbb{E}[x^2]<\infty$ and any $\tau\in(0,1)$,
\beq
\label{eq:PZ-general}
\Pr\bigl[x \ge \tau \mathbb{E}[x]\bigr]
\ge (1-\tau)^2 \frac{\bigl(\mathbb{E}[x]\bigr)^2}{\mathbb{E}[x^2]}.
\eeq
Applying this inequality to $x\equiv p_{\mG,\theta}(s)$ 
with $\tau=\frac{1}{2}$ 
yields
\begin{equation}
\Pr_{\mG,\theta}\Bigl[x\ge \frac{1}{2} 2^{-n}\Bigr]
\ \ge\ (1-\frac{1}{2})^2 
\frac{\bigl(\mathbb{E}_{\mG,\t}[x]\bigr)^2}{\mathbb{E}_{\mG,\t}[x^2]}
\ \ge\ \frac{1}{4 C(d)}.
\end{equation}

Thus we may take $a=\frac{1}{2}$ and $b=1/(4C(d))$ in \cref{eq:anticonc}, which depend only on $d$, and the claimed anticoncentration bound follows.  This is exactly the statement of \cref{conj:host-anticonc} specialized to $\mH_n=K_n$ and $\mathrm{Unif}[\mF_d(\mH)]=\mG_r(d,n)$. This concludes the proof of \cref{prop:Kn-anticonc}.

\textit{Proof of \cref{thm:host-sup}}.---%
Fix $0<\varepsilon<\varepsilon_\star$, where $\varepsilon_\star$ is a constant chosen below the anticoncentration-dependent threshold appearing in the standard IQP constant-error reductions (see, e.g.,~\cite{Bremner:2016aa,Hangleiter2018}). Consider one instance $(\mG,\theta)$ drawn from the QPU-restricted ensemble of \cref{def:QPU-R-E} with $\theta\sim\mathrm{Unif}[0,2\pi)^n$, and let $U^{(\theta)}_{\mathrm{DAD}}=U^{(\theta)}_{\mathrm{R}} U_{\mathrm{A}}(T) \UL$ be the corresponding device unitary with $\UL=S^\dagger W^{\otimes n}$ and $U^{(\theta)}_{\mathrm{R}}$ as in \cref{eq:RZ-UR}.

Choose smooth monotone schedules such that $
\eta\|\HI\| + K \beta \Delta\alpha  \le  \varepsilon/2$.
By \cref{lem:main} (applied with this $U^{(\theta)}_{\mathrm{R}}$), the output distribution of the device is within TV distance $\varepsilon/2$ of the IQP-type distribution with diagonal block $H'_Z(\mG,\theta)$ as in \cref{eq:HZ'}, i.e.,
\begin{equation}
\label{eq:TV-key-restate}
D (P_{U^{(\theta)}_{\mathrm{DAD}}}, P_{\UIQP^{(\theta)}})
 \le  \varepsilon/2 ,
\end{equation}
Any small programming errors in the angles can be absorbed here, as in \cref{eq:TV-key}. 

Suppose there exists a polynomial-time classical algorithm that, on a fraction at least $\delta_2$ of instances $(\mG,\theta)$, outputs samples from a distribution $P_{\mathrm{cl}}$ satisfying
$D (P_{\mathrm{cl}}, P_{U^{(\theta)}_{\mathrm{DAD}}})  <  \varepsilon/2$.
Combining this with \cref{eq:TV-key-restate} via the triangle inequality gives, for those instances,
\begin{equation}
\label{eq:cl-to-IQP}
D (P_{\mathrm{cl}}, P_{\UIQP^{(\theta)}})  <  \varepsilon .
\end{equation}

By \cref{conj:host-anticonc} or by \cref{prop:Kn-anticonc}, for each fixed output string $s\in\{0,1\}^n$ there exist constants $a,b>0$ (depending only on $d$) such that
$\Pr_{\mG,\theta} \bigl[P_{\UIQP^{(\theta)}}(s)\ge a 2^{-n}\bigr]\ \ge\ b$.

By \cref{conj:host-avgcase}, there exists a constant $\delta_1>0$ such that, unless $\mathrm{PH}$ collapses, no algorithm $A\in\mathrm{FBPP}^{\mathrm{NP}}$ can approximate $\mZ(\mG,\theta)$ within multiplicative error $1\pm 1/\mathrm{poly}(n)$ on more than a $\delta_1$-fraction of instances $(\mG,\theta)$ drawn from the QPU-restricted ensemble.

The mapping between $P_{\UIQP^{(\theta)}}(s)$ and $\mZ(\mG,\theta)$ is the standard one for IQP circuits with commuting diagonal gates, so an $\varepsilon$-accurate sampler satisfying \cref{eq:cl-to-IQP} on a fraction $\gamma$ of instances would, via Stockmeyer counting and the anticoncentration bound above, yield such an $\mathrm{FBPP}^{\mathrm{NP}}$ algorithm on at least a constant fraction $c\,\gamma$ of those instances, where $c>0$ depends only on $a$ and $b$ (and hence only on $d$). 

Choosing $\varepsilon_\star>0$ small enough relative to the anticoncentration constant $a$ [so that an additive approximation of the probabilities $p_{\mG,\theta}(s)$ can be translated into a multiplicative approximation of the partition function $\mZ(\mG,\theta)$] ensures that this Stockmeyer-based reduction applies.
Let $\delta_2 := \delta_1 / c$. Then any polynomial-time classical sampler achieving TV error $<\varepsilon$ (for some $\varepsilon$ with $0< \varepsilon < \varepsilon_\star$) on a fraction $\gamma>\delta_2$ of instances $(\mG,\theta)$ would give an $\mathrm{FBPP}^{\mathrm{NP}}$ algorithm that approximates $\mZ(\mG,\theta)$ on more than a $\delta_1$-fraction of instances, contradicting \cref{conj:host-avgcase} unless $\mathrm{PH}$ collapses to the third level $\Sigma_3$. This concludes the proof.

We note that the plausibility of \cref{conj:host-avgcase} is supported by \emph{worst}-case $\#$P-hardness results for complex Ising
models on bounded-degree graphs~\cite{GoldbergGuo2017ComplexIsing,Galanis2022ComplexIsingBoundedDegree}. To prove \cref{conj:host-avgcase} would require an \emph{average}-case hardness theorem for the complex-valued Ising model on simple $D$-regular graphs, showing that a multiplicative approximation to $\mZ(\mG,\theta)$ is \#P-hard on a  constant fraction of $(\mG,\theta)$.

\end{document}